\pdfoutput=1
\documentclass[11pt]{article}

\usepackage[preprint]{acl}

\usepackage{times}
\usepackage{latexsym}

\usepackage[T1]{fontenc}
\usepackage[utf8]{inputenc}

\usepackage{microtype}

\usepackage{inconsolata}

\usepackage{graphicx}

\usepackage{url}
\usepackage{booktabs}
\usepackage{amsfonts}
\usepackage{nicefrac}
\usepackage{xcolor}
\usepackage{fancyvrb}
\usepackage{fvextra}

\usepackage{wrapfig}
\usepackage{tikz}
\usepackage{float}
\usepackage{caption}
\usepackage{enumerate}
\usepackage{array}
\usepackage[ruled, linesnumbered, boxed]{algorithm2e}
\usepackage{todonotes}

\usepackage{multirow}
\usepackage{makecell}
\usepackage{bbding}
\usepackage{mathtools}

\usepackage[nameinlink]{cleveref}
\Crefname{figure}{Figure}{Figures}
\crefname{figure}{Figure}{Figures}
\crefname{example}{Example}{Example}
\crefname{theorem}{Theorem}{Theorem}
\crefname{corollary}{Corollary}{Corollary}
\crefname{lemma}{Lemma}{Lemma}
\crefname{proposition}{Proposition}{Proposition}
\crefname{assumption}{Assumption}{Assumption}
\crefname{section}{Section}{Section}
\crefname{algorithm}{Algorithm}{Algorithm}

\usepackage{amsthm,thmtools}
\declaretheorem[name=Theorem,numberwithin=section]{theorem}
\declaretheorem[name=Definition,style=definition]{definition}

\declaretheorem[name=Proposition,numberlike=theorem]{proposition}

\title{Reducing Detail Hallucinations in Long-Context Regulatory \\ Understanding via Targeted Preference Optimization}

\author{Yang Liu$^1$, Bin Chong$^{2}$\thanks{Corresponding author.}, Yuhan Lin$^3$, Chongyang Zhang$^4$, Hao Zheng$^4$, Ziyi Zhang$^2$, \\
\textbf{Jiayu Liang}$^5$, \textbf{Ran Ran}$^6$,
\textbf{Qian Li}$^7$, \textbf{Kefu Xu}$^2$\\
  $^1$Tsinghua University
  $^2$Peking University
  $^3$Fudan University\\
  $^4$Fulai Shuchuang (Beijing) Intelligent Technology Co., Ltd
  $^5$Soochow University\\
  $^6$University of California, Berkeley
  $^7$Beijing University of Posts and Telecommunications\\
  }

\begin{document}
\maketitle
\begin{abstract}
Large language models (LLMs) frequently produce \emph{detail hallucinations} when processing long regulatory documents, including subtle errors in threshold values, units, scopes, obligation levels, and conditions that preserve surface plausibility while corrupting safety-critical parameters. We formalize this phenomenon through a fine-grained \emph{Detail Error Taxonomy} of five error types and introduce \textbf{DetailBench}, a benchmark built from 172 real regulatory documents and 150 synthetic documents spanning three jurisdictions, with human-annotated detail-level ground truth comprising 13,000 preference pairs. We propose \textbf{DetailDPO}, a targeted preference optimization framework that constructs contrastive pairs differing in exactly one detail dimension, concentrating DPO gradient signal on detail-bearing~tokens. We provide theoretical analysis showing why \emph{minimal detail perturbation} pairs yield gradient concentration under mild assumptions. Experiments on the Qwen2.5 family (7B, 14B, 72B) and Llama-3.1-8B across three context-length tiers (8K--64K tokens) show that DetailDPO reduces the Detail Error Rate by 42--61\% relative to baselines, with consistent gains across all five error types and cross-domain transfer to financial and medical documents.
\end{abstract}

\section{Introduction}

Large language models (LLMs) have demonstrated remarkable capabilities in text understanding and generation~\cite{brown2020language,ouyang2022training}, including promising applications in legal AI~\cite{katz2017general,chalkidis2020legal} and long-document comprehension~\cite{beltagy2020longformer}. However, a critical yet under-studied failure mode emerges when LLMs process long regulatory and technical documents: \emph{detail hallucination}. In this failure mode, the model produces outputs that are broadly correct in structure and narrative but contain subtle errors in safety-critical details such as threshold numbers, measurement units, regulatory scopes, obligation levels (e.g., ``shall'' vs.\ ``should''), and conditional applicability. Unlike coarse-grained hallucination where the model fabricates entirely unsupported claims, detail hallucination preserves surface plausibility while silently corrupting parameters that determine real-world compliance outcomes, making it especially~insidious.

In compliance-sensitive domains, including energy regulation~\cite{noyan2023hydrogen}, pharmaceutical safety, and financial oversight, a single misquoted pressure threshold or a confused unit conversion can lead to catastrophic safety failures or costly legal violations. Yet existing faithfulness evaluation and mitigation methods predominantly target coarse-grained factual consistency~\cite{ji2023survey,maynez2020faithfulness}, leaving fine-grained detail errors largely unaddressed. As context lengths grow (modern LLMs support 32K--128K+ tokens), the problem intensifies: longer inputs contain more detail-bearing tokens competing for model attention, and empirical evidence suggests that detail accuracy degrades faster than overall task performance as context length~increases.

Systematically addressing detail hallucination in long-context settings presents three key challenges. \textbf{(C1)~Lack of fine-grained error characterization:} Existing hallucination taxonomies do not distinguish between types of detail errors (numeric, unit, scope, level, condition), hindering targeted mitigation. \textbf{(C2)~Absence of detail-level benchmarks:} Current long-context benchmarks~\cite{bai2024longbench,shaham2022scrolls} evaluate overall comprehension but do not isolate detail-level faithfulness, making it difficult to measure progress. \textbf{(C3)~No alignment methods targeting detail faithfulness:} Standard preference optimization treats all errors equally; there is no mechanism to concentrate learning signal on detail-bearing tokens where errors are most consequential.

To address these gaps, we make the following contributions:
\begin{enumerate}
    \item \textbf{Detail Error Taxonomy and DetailBench.} We propose a principled taxonomy of five detail error types grounded in regulatory compliance practice, and construct \textbf{DetailBench}, a benchmark built from real regulatory documents (Chinese GB standards, US~CFR, EUR-Lex) with human-annotated detail-level ground truth spanning multiple context lengths (8K--64K tokens). (\textbf{C1, C2})
    \item \textbf{DetailDPO Framework.} We introduce a targeted preference optimization approach that constructs \emph{minimal detail perturbation} pairs, where chosen and rejected responses differ only in a single detail dimension, and show theoretically and empirically that this concentrates DPO gradient on detail-bearing~tokens. (\textbf{C3})
    \item \textbf{Theoretical Analysis.} We formalize the Detail Error Rate (DER) metric and prove that minimal-perturbation DPO yields gradient concentration on detail tokens under mild assumptions, providing a principled explanation for why targeted preference pair construction outperforms generic~alternatives.
    \item \textbf{Comprehensive Experiments.} We evaluate across the Qwen2.5 family (7B/14B/72B) and Llama-3.1-8B, three context lengths (8K--64K), five alignment baselines (SFT, RLHF, DPO, KTO, SimPO), and three domains (regulatory, financial, medical), demonstrating consistent detail error reduction and cross-domain~transferability.
\end{enumerate}

\section{Problem Formulation}

\subsection{Task Definition}

Given a set of regulatory documents $\mathcal{D} = \{d_1, d_2, \ldots, d_n\}$, where each document $d_i$ contains segments $\{s_{i,1}, s_{i,2}, \ldots, s_{i,m_i}\}$, and a compliance query $q$, the task is to produce a structured compliance analysis $\mathcal{A} = (y, \mathcal{C}, \mathcal{R}, \mathcal{E})$ comprising: compliance judgment $y \in \{\text{compliant}, \text{non-compliant}\}$, extracted constraints $\mathcal{C} = \{c_1, \ldots, c_k\}$ each with associated detail parameters, risk assessment $\mathcal{R}$, and evidence citations $\mathcal{E} = \{(d_i, s_{i,j}, \text{quote})\}$ binding each conclusion to source passages.

The total input length $L = |q| + \sum_{i=1}^n \sum_{j=1}^{m_i} |s_{i,j}|$ ranges from 8K to 64K+ tokens in realistic multi-document settings.

\subsection{Detail Error Taxonomy}
\label{sec:taxonomy}

We identify five types of \emph{detail errors} that LLMs commit when extracting information from regulatory documents. Unlike coarse hallucination---which fabricates entire unsupported claims---each detail error preserves the overall claim structure while silently corrupting a single safety-critical parameter. The five types span the principal axes along which a numeric or textual constraint can be misrepresented: \textbf{threshold errors} alter the numeric limit itself; \textbf{unit errors} swap the measurement unit without adjusting the value, producing a physically incompatible quantity; \textbf{scope errors} expand or contract the set of entities or situations to which a requirement applies; \textbf{level errors} change the obligation strength (e.g., mandatory $\to$ advisory), altering the legal force of the requirement; and \textbf{condition errors} omit or modify the conditional clause that gates when a requirement is triggered, potentially rendering an unconditional rule as conditional or vice versa. Formally:

\begin{definition}[Detail Error Types]
\label{def:error_types}
Let $c^* = (b, u, \sigma, \ell, \phi)$ be a ground-truth constraint with threshold $b$, unit $u$, scope $\sigma$, obligation level $\ell$, and conditions $\phi$. A model prediction $\hat{c} = (\hat{b}, \hat{u}, \hat{\sigma}, \hat{\ell}, \hat{\phi})$ commits:
\begin{enumerate}
    \item \textbf{Threshold error} ($\tau_1$): $\hat{b} \neq b$ (e.g., ``70\,MPa'' $\to$ ``75\,MPa'')
    \item \textbf{Unit error} ($\tau_2$): $\hat{u} \neq u$ without value conversion (e.g., ``70\,MPa'' $\to$ ``70\,bar'')
    \item \textbf{Scope error} ($\tau_3$): $\hat{\sigma} \neq \sigma$ (e.g., ``stationary storage'' $\to$ ``all storage systems'')
    \item \textbf{Level error} ($\tau_4$): $\hat{\ell} \neq \ell$ (e.g., ``shall'' $\to$ ``should'')
    \item \textbf{Condition error} ($\tau_5$): $\hat{\phi} \neq \phi$ (e.g., dropping the condition ``when pressure $> 50$\,MPa'')
\end{enumerate}
\end{definition}

These five types are mutually exclusive by construction: each targets a distinct slot in the constraint tuple $c^* = (b, u, \sigma, \ell, \phi)$, so a single prediction can independently commit errors of multiple types. In practice, threshold ($\tau_1$) and unit ($\tau_2$) errors tend to be the most localized---corrupting a single numeric token or unit string---and therefore the most amenable to targeted correction. Scope ($\tau_3$), level ($\tau_4$), and condition ($\tau_5$) errors typically span multiple tokens and arise from misinterpretation of clause structure rather than simple numeric substitution, making them subtler and often more consequential for compliance~outcomes.

\Cref{tab:error_type_examples} provides concrete examples of each error type drawn from regulatory compliance scenarios.

\begin{table*}[t]
\centering
\small
\begin{tabular}{llll}
\toprule
\textbf{Type} & \textbf{Ground Truth} & \textbf{Model Output (Error)} & \textbf{What Changed} \\
\midrule
$\tau_1$ Threshold & $\leq$ \textbf{70.0}\,MPa & $\leq$ \textbf{75.0}\,MPa & Numeric value altered \\
$\tau_2$ Unit & 70.0\,\textbf{MPa} & 70.0\,\textbf{bar} & Unit swapped, no conversion \\
$\tau_3$ Scope & \textbf{stationary} storage & \textbf{all} storage systems & Scope broadened \\
$\tau_4$ Level & operators \textbf{shall} hold cert. & operators \textbf{should} hold cert. & Obligation downgraded \\
$\tau_5$ Condition & \textbf{when pressure $>$ 50\,MPa} & \emph{(condition omitted)} & Condition dropped \\
\bottomrule
\end{tabular}%
\caption{Concrete examples of each detail error type. Bold indicates the corrupted element. These errors preserve overall claim plausibility while altering safety-critical parameters.}
\label{tab:error_type_examples}
\end{table*}

\subsection{Detail Error Rate (DER)}
\label{sec:der}

We formalize a metric to measure detail faithfulness. Each constraint $c_k \in \mathcal{C}$ contains five detail elements corresponding to the components $(b, u, \sigma, \ell, \phi)$; we denote the $k$-th ground-truth detail element of sample $i$ as $d_{i,k}^*$ and its prediction as $\hat{d}_{i,k}$.

\begin{definition}[Detail Error Rate]
Given $N$ test samples, where sample $i$ contains $K_i \geq 0$ detail-bearing elements and the total count $K_{\text{tot}} = \sum_{i=1}^N K_i > 0$, define:
\begin{equation}
\text{DER} = \frac{1}{K_{\text{tot}}} \sum_{i=1}^{N} \sum_{k=1}^{K_i} \mathbb{1}\!\left[\hat{d}_{i,k} \neq d_{i,k}^*\right]
\end{equation}
where $d_{i,k}^*$ is the $k$-th ground-truth detail element and $\hat{d}_{i,k}$ the corresponding prediction. The type-specific DER for error type $\tau_j$ (defined when $|\mathcal{S}_j| > 0$) is:
\begin{equation}
\text{DER}_{\tau_j} = \frac{1}{|\mathcal{S}_j|} \sum_{(i,k) \in \mathcal{S}_j} \mathbb{1}\!\left[\hat{d}_{i,k} \neq d_{i,k}^*\right]
\end{equation}
where $\mathcal{S}_j$ is the set of detail elements of type $\tau_j$.
\end{definition}

The matching predicate $\hat{d}_{i,k} \neq d_{i,k}^*$ is type-dependent: for threshold errors ($\tau_1$), we require exact numeric match after parsing to a canonical float representation (e.g., ``70.0'' $\neq$ ``75.0''); for unit errors ($\tau_2$), exact unit string match after canonicalization (e.g., ``MPa'' $\neq$ ``bar''). For scope ($\tau_3$), level ($\tau_4$), and condition ($\tau_5$) errors, we first normalize by lowercasing and removing punctuation, then compute token-level F1 between the predicted and ground-truth spans; a match is accepted when F1 $\geq 0.8$. Ambiguous cases (F1 in $[0.7, 0.9]$) undergo manual adjudication.

DER complements standard metrics (compliance accuracy, evidence F1) by isolating the fine-grained faithfulness~dimension.

\section{DetailDPO Framework}
\label{sec:method}

The DetailDPO framework consists of three components: (1) data construction combining real and synthetic regulatory documents, (2) minimal detail perturbation for preference pair generation, and (3) DPO training with detail-aware gradient concentration.

\begin{figure*}[t]
\centering
\includegraphics[width=0.95\linewidth]{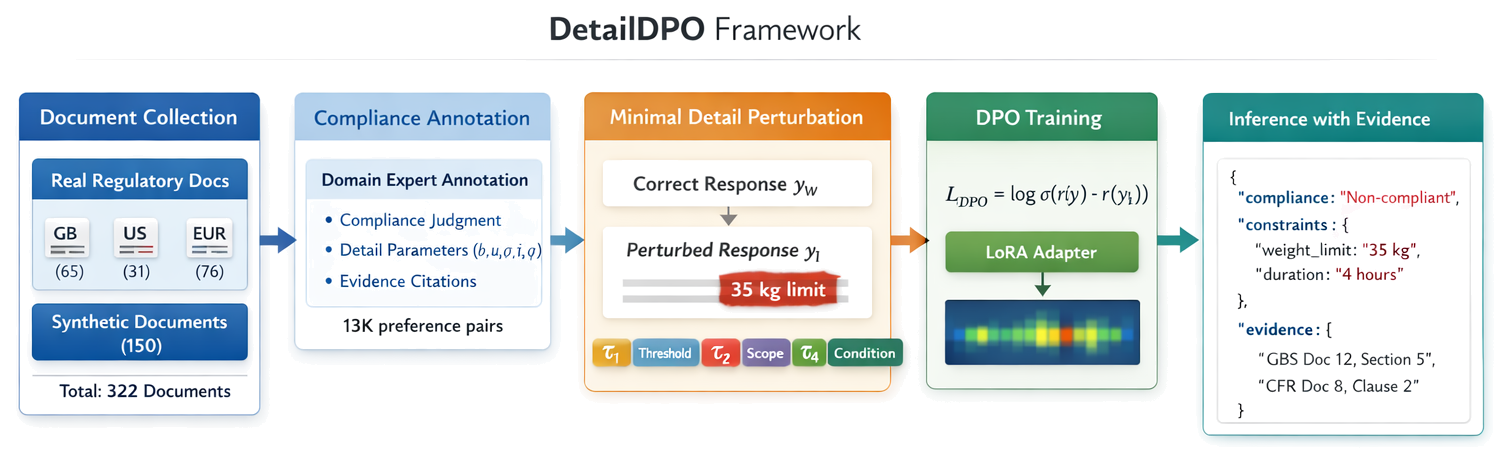}
\caption{DetailDPO pipeline: (1) Real + synthetic regulatory document collection, (2) Compliance annotation with detail-level ground truth, (3) Minimal detail perturbation to generate preference pairs for 5 error types, (4) DPO training with LoRA, (5) Evidence citation mechanism binding conclusions to source segments.}
\label{fig:pipeline}
\end{figure*}

\Cref{fig:pipeline} illustrates the complete pipeline from data construction through training to~inference.

\subsection{Data Construction}
\label{sec:data}

\paragraph{Real Regulatory Documents.} We collect 172 documents from three official sources: (1) 65 Chinese GB/T standards on hydrogen production, storage, transportation, and safety, reconstructed from publicly available metadata and clause structures published by the National Standard system; (2) 31 US Code of Federal Regulations (CFR) parts via the eCFR REST API, focusing on Title~49 (Transportation, Parts~171--199) and Title~40 (Environmental Protection, Parts~60--268); (3) 76 EU regulations retrieved through the Publications Office CELLAR SPARQL endpoint, covering hydrogen infrastructure, pressure equipment, and environmental directives. These provide authoritative text with traceable~identifiers.

\paragraph{Synthetic Augmentation.} To increase training data volume and ensure balanced coverage of all error types, we supplement with synthetic documents generated from domain templates. For each synthetic document $d_i$, we sample a topic domain, generate structured segments with threshold constraints (sampled with probability 0.7), requirement levels, and conditional clauses. Crucially, synthetic documents mirror the structure and terminology of real standards to support~transfer.

\paragraph{Compliance Annotation.} For each (document set, query) pair, domain experts annotate: compliance judgment, extracted constraints with all five detail parameters $(b, u, \sigma, \ell, \phi)$, and evidence citations linking each constraint to source segments. For synthetic data, annotations are generated programmatically and verified by~sampling.

\subsection{Minimal Detail Perturbation}
\label{sec:perturbation}

The key insight of DetailDPO is the construction of preference pairs where the rejected response differs from the chosen response in \emph{exactly one detail dimension}:

\begin{definition}[Minimal Detail Perturbation]
\label{def:mdp}
Given a correct response $y_w$ containing detail elements $\{d_1^*, \ldots, d_K^*\}$, a minimal detail perturbation of type $\tau_j$ produces $y_l$ by replacing exactly one detail element $d_k^*$ of type $\tau_j$ with an incorrect value $\tilde{d}_k$, while keeping all other tokens identical:
\begin{equation}
y_l = y_w[\,d_k^* \leftarrow \tilde{d}_k\,], \quad \text{type}(d_k^*) = \tau_j
\end{equation}
\end{definition}

Concretely, for threshold errors ($\tau_1$), $\tilde{d}_k$ is produced by multiplying the original value by a factor drawn from $[0.8,0.9]\cup[1.1,1.2]$, yielding a plausible yet clearly incorrect perturbation. For unit errors ($\tau_2$), the unit symbol is replaced by a compatible but distinct unit (e.g., MPa $\to$ bar) while the numeric value remains unchanged. Scope ($\tau_3$) and condition ($\tau_5$) perturbations broaden or narrow the relevant phrase, or drop the conditional clause entirely. Level ($\tau_4$) perturbations substitute the obligation keyword (e.g., ``shall'' $\to$ ``should'', ``must'' $\to$ ``may''). \Cref{alg:perturbation} in the Appendix gives the complete procedure for all five types.

The defining property of this construction is that $y_w$ and $y_l$ share every token except those at the perturbed detail positions $\mathcal{P}$. Consequently, the preference signal is concentrated on exactly one error dimension, with no confounding signal from unrelated token differences. We formalize why this leads to gradient concentration on detail-bearing tokens in the~following.

\subsection{DPO Training with Gradient Concentration}
\label{sec:dpo_training}

We train with the standard DPO objective:
\begin{equation}
\mathcal{L}_{\text{DPO}} = -\mathbb{E}_{(x,y_w,y_l)}\!\Bigl[\log\sigma\bigl(\beta\,\Delta(x,y_w,y_l)\bigr)\Bigr]
\label{eq:dpo}
\end{equation}
where $\Delta(x,y_w,y_l) = \log\frac{\pi_\theta(y_w\mid x)}{\pi_{\mathrm{ref}}(y_w\mid x)} - \log\frac{\pi_\theta(y_l\mid x)}{\pi_{\mathrm{ref}}(y_l\mid x)}$, $\pi_\theta$ is the policy, $\pi_{\mathrm{ref}}$ the reference model, and $\beta$ the temperature.

The critical property of minimal detail perturbation is that $y_w$ and $y_l$ share all tokens except at the perturbed detail positions. Let $\mathcal{P} \subset \{1,\ldots,T\}$ denote the token positions where $y_w$ and $y_l$ differ (the \emph{detail tokens}), and $\bar{\mathcal{P}}$ the shared positions.

\begin{proposition}[Gradient Concentration on Detail Tokens]
\label{prop:gradient}
Under the DPO objective with minimal detail perturbation pairs $(y_w, y_l)$, let $\Delta_t = \nabla_\theta \log \pi_\theta(y_{w,t}|y_{w,<t}, x) - \nabla_\theta \log \pi_\theta(y_{l,t}|y_{l,<t}, x)$ be the per-token gradient difference. Assume:
\begin{enumerate}
    \item[(A1)] $y_w$ and $y_l$ agree on all tokens before position $\min(\mathcal{P})$;
    \item[(A2)] \textbf{($\epsilon$-bounded sensitivity)} For all $t \in \bar{\mathcal{P}}$ with $t > \min(\mathcal{P})$, $\|\Delta_t\| \leq \epsilon$ for some $\epsilon > 0$;
    \item[(A3)] There exists $\delta > 0$ such that $\|\Delta_t\| \geq \delta$ for all $t \in \mathcal{P}$.
\end{enumerate}
Then the contribution of non-detail tokens to the total gradient is bounded:
\begin{equation}
\left\| \sum_{t \in \bar{\mathcal{P}}} \Delta_t \right\| \leq |\bar{\mathcal{P}}| \cdot \epsilon
\end{equation}
while $\left\| \sum_{t \in \mathcal{P}} \Delta_t \right\| \geq \sqrt{|\mathcal{P}|} \cdot \delta$ in general (by triangle inequality), and $\geq |\mathcal{P}| \cdot \delta$ when the per-position gradients within $\mathcal{P}$ are aligned. The gradient concentrates on detail tokens whenever $|\bar{\mathcal{P}}| \cdot \epsilon < \sqrt{|\mathcal{P}|} \cdot \delta$.
\end{proposition}

\begin{proof}[Proof sketch]
The DPO gradient involves $\nabla_\theta [r_\theta(x,y_w) - r_\theta(x,y_l)]$ where $r_\theta(x,y) = \sum_t \log \pi_\theta(y_t|y_{<t},x) - \log \pi_{\text{ref}}(y_t|y_{<t},x)$. By (A1), for tokens $t < \min(\mathcal{P})$, both the token identity and conditioning context are identical, so $\Delta_t = 0$ exactly. For $t \in \bar{\mathcal{P}}$ with $t > \min(\mathcal{P})$, $y_{w,t} = y_{l,t}$ but the conditioning contexts differ through the perturbed tokens; (A2) bounds these contributions. At positions $t \in \mathcal{P}$, the token identities differ, and (A3) provides a lower bound on the gradient magnitude. The full proof in the Appendix additionally justifies (A2) via Lipschitz continuity of the softmax-attention mechanism under bounded model weights and input norms. See \Cref{sec:gradient_empirical} for empirical verification.
\end{proof}

This result provides a theoretical justification for why DetailDPO is more effective than generic DPO (with arbitrary rejected responses) at reducing detail errors: the learning signal is concentrated precisely where it~matters.

\subsection{Evidence Citation Mechanism}

To ensure traceability, the model outputs structured JSON where each constraint, risk, or recommendation binds to evidence citations $\{(d_i, s_{i,j}, \text{quote})\}$. We enforce this during training by augmenting the DPO loss with an evidence supervision term:
\begin{equation}
\mathcal{L}_{\text{total}} = \mathcal{L}_{\text{DPO}} + \lambda \cdot \mathcal{L}_{\text{evid}}
\end{equation}
where $\mathcal{L}_{\text{evid}} = -\frac{1}{|\mathcal{E}|}\sum_{e \in \mathcal{E}} \log \pi_\theta(e_{\text{quote}} | x, e_{\text{doc}}, e_{\text{seg}})$ is the negative log-likelihood of generating correct evidence quotes given the document and segment identifiers, and $\lambda = 0.5$ balances detail faithfulness and citation quality. Evidence consistency is evaluated by checking whether the quoted text appears verbatim (exact match) or is semantically equivalent (measured by BERTScore~\cite{zhang2020bertscore} $\geq 0.9$) to the source segment.

\section{DetailBench: A Detail Faithfulness Benchmark}
\label{sec:benchmark}

\subsection{Data Sources and Collection}

DetailBench is constructed from 172 real regulatory documents and 150 synthetic documents (322 total):

\begin{itemize}
    \item \textbf{GB Standards} (Chinese national standards): 65 documents covering hydrogen production, storage, transportation, and safety (GB/T series). Due to access restrictions on the full-text system (\url{openstd.samr.gov.cn}), a subset of standards is reconstructed from publicly available metadata and clause structures.
    \item \textbf{US CFR} (Code of Federal Regulations via the eCFR API): 31 documents from Title~49 (Transportation, Parts~171--199) and Title~40 (Environmental Protection, Parts~60--268), focusing on hazardous materials, pipeline safety, air emissions, and hazardous waste regulations.
    \item \textbf{EUR-Lex} (EU regulations via the CELLAR API): 76 documents on hydrogen infrastructure, clean energy, pressure equipment, and environmental protection, including AFIR, RED~II/III, EU~ETS, and Seveso~III.
    \item \textbf{Synthetic}: 150 documents generated from domain templates to augment training data with balanced error-type coverage.
\end{itemize}

\subsection{Annotation Process}

Three domain experts with regulatory compliance experience independently annotated each sample. For each (document set, compliance query) pair, annotators identify: (1) the correct compliance judgment, (2) all detail-bearing elements with their types ($\tau_1$--$\tau_5$) and ground-truth values, and (3) evidence citations linking each element to source segments with verbatim quotes. Inter-annotator agreement (Cohen's $\kappa$) is 0.81 for compliance judgment and 0.76 for detail element extraction, indicating substantial~agreement.

\subsection{Benchmark Statistics}

\begin{figure}[t]
\centering
\includegraphics[width=0.95\columnwidth]{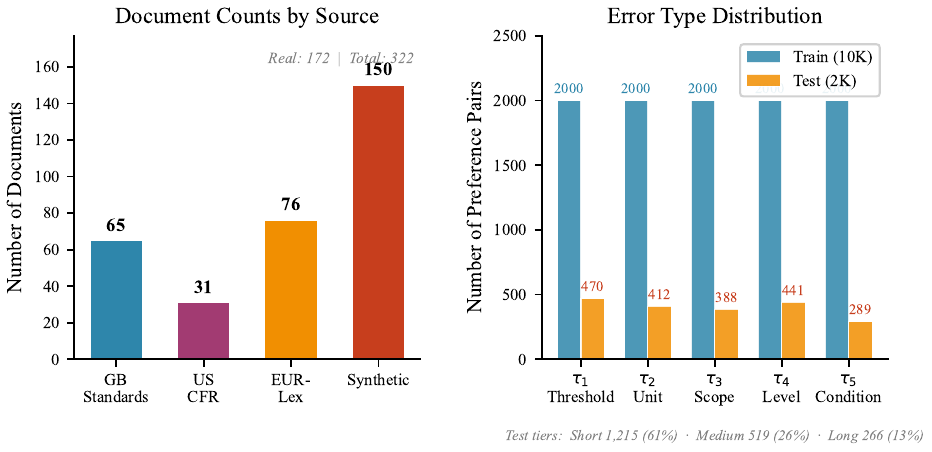}
\caption{DetailBench statistics. \textbf{Left}: Document counts by source (GB: 65, CFR: 31, EUR-Lex: 76, Synthetic: 150). \textbf{Right}: Distribution of detail elements across five error types and three context-length tiers.}
\label{fig:data_statistics}
\end{figure}

From the 322 source documents, we create multiple (document set, compliance query) samples by varying the document combinations and query formulations; each sample then yields up to five preference pairs (one per error type via minimal detail perturbation). We enforce strict document-level splits: documents used in test samples do not appear in training or validation, preventing any information leakage. This combinatorial process produces 13,000 preference pairs in total: 10,000 for training, 1,000 for validation, and 2,000 for testing. Test samples are constructed at three context-length tiers: \textbf{Short} (8K--16K tokens, 1--3 documents), \textbf{Medium} (16K--32K tokens, 3--8 documents), and \textbf{Long} (32K--64K tokens, 8--15 documents). Due to the combinatorial difficulty of assembling 8--15 documents into the 32K--64K range, the Long tier contains fewer samples (13\% of test) than Short (61\%) and Medium (26\%); we report per-tier results to avoid conflating this imbalance with model performance. The five error types are balanced at 20\% each in the preference training data. \Cref{fig:data_statistics} summarizes the benchmark statistics.

\section{Experimental Setup}

\subsection{Base Models}

To evaluate generalizability across model scales and families, we experiment with:
\begin{itemize}
    \item \textbf{Qwen2.5-7B-Instruct}~\cite{qwen2024qwen25}: 7B parameters, 32K context window.
    \item \textbf{Qwen2.5-14B-Instruct}~\cite{qwen2024qwen25}: 14B parameters, 32K context window.
    \item \textbf{Qwen2.5-72B-Instruct}~\cite{qwen2024qwen25}: 72B parameters, 128K context window.
    \item \textbf{Llama-3.1-8B-Instruct}~\cite{grattafiori2024llama3}: 8B parameters, 128K context window.
\end{itemize}
DetailDPO training uses LoRA~\cite{hu2022lora} for parameter-efficient fine-tuning on all models.

\subsection{Baselines}

We compare against nine baselines spanning zero-shot prompting, supervised fine-tuning, four preference optimization variants (RLHF, Generic DPO, KTO, SimPO), retrieval-augmented generation, and commercial API models (GPT-4o, Claude-3.5-Sonnet). Detailed descriptions of each baseline are provided in \Cref{sec:baselines_appendix}.

\subsection{Training Configuration}

\begin{table}[t]
\centering
\small
\begin{tabular}{ll}
\toprule
\textbf{Hyperparameter} & \textbf{Value} \\
\midrule
Learning Rate & $1.0 \times 10^{-5}$ \\
Training Epochs & 3 \\
Batch Size & 2 \\
Gradient Accumulation Steps & 4 \\
DPO Temperature $\beta$ & 0.1 \\
LoRA Rank $r$ & 16 \\
LoRA Alpha & 32 \\
LoRA Dropout & 0.1 \\
Max Sequence Length & 32,768 \\
\bottomrule
\end{tabular}
\caption{Training hyperparameters for DetailDPO.}
\label{tab:hyperparameters}
\end{table}

\Cref{tab:hyperparameters} summarizes the key hyperparameters. Further details on LoRA target modules and computational resources are provided in \Cref{sec:training_details}.

\subsection{Evaluation Protocol}

Each test sample inputs 1--15 documents (depending on context-length tier). Models produce structured JSON containing compliance judgments, constraint extractions with detail parameters, and evidence citations. We evaluate along four dimensions:
\begin{itemize}
    \item \textbf{Compliance Accuracy}: fraction of correct compliance judgments.
    \item \textbf{Detail Error Rate (DER)}: as defined in \Cref{sec:der}, both overall and per-type.
    \item \textbf{Evidence F1}: precision, recall, and F1 of predicted evidence citations against ground truth.
    \item \textbf{Evidence Consistency}: fraction of citations where the quoted text matches the source (exact or BERTScore $\geq 0.9$).
\end{itemize}
Additionally, we conduct: (1) \textbf{context-length analysis} across the three tiers, (2) \textbf{cross-domain transfer} by training on regulatory data and evaluating on financial (SEC filings) and medical (clinical guideline) documents, and (3) \textbf{gradient analysis} to empirically verify \Cref{prop:gradient}.

\section{Results and Analysis}

\subsection{Main Results}

\Cref{tab:main_results} presents the main results on DetailBench using Qwen2.5-7B as the primary base model. \Cref{fig:results_comparison} visualizes performance across methods.

\begin{figure*}[t]
\centering
\includegraphics[width=0.95\linewidth]{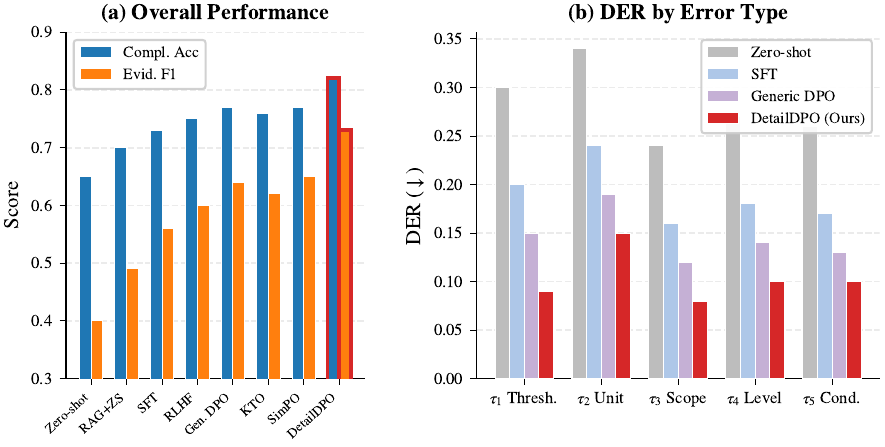}
\caption{Results overview. \textbf{Left}: Compliance accuracy and evidence F1 across methods. \textbf{Right}: DER by error type, showing DetailDPO's consistent reduction across all five categories.}
\label{fig:results_comparison}
\end{figure*}

\begin{table*}[t]
\centering
\small
\resizebox{\linewidth}{!}{%
\begin{tabular}{lcccccc}
\toprule
\textbf{Method} & \textbf{Compl.\ Acc.} & \textbf{DER$\downarrow$} & \textbf{Evid.\ Prec.} & \textbf{Evid.\ Rec.} & \textbf{Evid.\ F1} & \textbf{Evid.\ Consist.} \\
\midrule
\multicolumn{7}{l}{\emph{Commercial API baselines (zero-shot)}} \\
GPT-4o & 0.78 & 0.16 & 0.68 & 0.61 & 0.64 & 0.58 \\
Claude-3.5-Sonnet & 0.76 & 0.18 & 0.65 & 0.59 & 0.62 & 0.55 \\
\midrule
\multicolumn{7}{l}{\emph{Qwen2.5-7B-Instruct}} \\
Zero-shot & 0.65 & 0.28 & 0.42 & 0.38 & 0.40 & 0.35 \\
RAG + Zero-shot & 0.70 & 0.23 & 0.52 & 0.47 & 0.49 & 0.44 \\
SFT & 0.73\tiny{$\pm$.02} & 0.19\tiny{$\pm$.01} & 0.58\tiny{$\pm$.02} & 0.55\tiny{$\pm$.02} & 0.56\tiny{$\pm$.02} & 0.51\tiny{$\pm$.02} \\
RLHF (PPO) & 0.75\tiny{$\pm$.02} & 0.16\tiny{$\pm$.01} & 0.62\tiny{$\pm$.01} & 0.59\tiny{$\pm$.02} & 0.60\tiny{$\pm$.01} & 0.55\tiny{$\pm$.02} \\
Generic DPO & 0.77\tiny{$\pm$.01} & 0.15\tiny{$\pm$.01} & 0.65\tiny{$\pm$.01} & 0.63\tiny{$\pm$.01} & 0.64\tiny{$\pm$.01} & 0.59\tiny{$\pm$.01} \\
KTO & 0.76\tiny{$\pm$.01} & 0.16\tiny{$\pm$.01} & 0.63\tiny{$\pm$.02} & 0.61\tiny{$\pm$.01} & 0.62\tiny{$\pm$.01} & 0.57\tiny{$\pm$.02} \\
SimPO & 0.77\tiny{$\pm$.01} & 0.14\tiny{$\pm$.01} & 0.66\tiny{$\pm$.01} & 0.64\tiny{$\pm$.01} & 0.65\tiny{$\pm$.01} & 0.60\tiny{$\pm$.01} \\
\textbf{DetailDPO (Ours)} & \textbf{0.82}\tiny{$\pm$.01} & \textbf{0.11}\tiny{$\pm$.01} & \textbf{0.75}\tiny{$\pm$.01} & \textbf{0.72}\tiny{$\pm$.01} & \textbf{0.73}\tiny{$\pm$.01} & \textbf{0.68}\tiny{$\pm$.01} \\
\bottomrule
\end{tabular}%
}
\caption{Main results on DetailBench (Qwen2.5-7B base, averaged over all context-length tiers). DER = overall Detail Error Rate (lower is better). Best open-source results in \textbf{bold}. $\pm$ denotes std.\ over 3 runs with different random seeds. Commercial API results are deterministic (temperature=0).}
\label{tab:main_results}
\end{table*}

DetailDPO achieves the best results among all open-source methods, with DER of 0.11 vs.\ 0.28 for zero-shot (61\% relative reduction) and 0.19 for SFT (42\% reduction). Notably, DetailDPO outperforms Generic DPO (0.15 DER) by 27\%, demonstrating that the minimal detail perturbation strategy is critical: not just preference optimization in general, but the \emph{structure} of the preference pairs matters. DetailDPO also surpasses the zero-shot performance of GPT-4o (0.16 DER) and Claude-3.5-Sonnet (0.18 DER) in detail faithfulness despite being a 7B model, suggesting that targeted alignment can compensate for scale advantages in this specific dimension. We note that commercial models with task-specific prompting (e.g., chain-of-thought or few-shot) may narrow this~gap.

\subsection{Per-Type Detail Error Analysis}

\Cref{tab:error_analysis} breaks down DER by error type, revealing which detail dimensions benefit most from DetailDPO.

\begin{table}[t]
\centering
\small
\resizebox{\linewidth}{!}{%
\begin{tabular}{lccccc}
\toprule
\textbf{Method} & \textbf{$\tau_1$ Thresh.} & \textbf{$\tau_2$ Unit} & \textbf{$\tau_3$ Scope} & \textbf{$\tau_4$ Level} & \textbf{$\tau_5$ Cond.} \\
\midrule
Zero-shot & 0.30 & 0.34 & 0.24 & 0.27 & 0.26 \\
SFT & 0.20 & 0.24 & 0.16 & 0.18 & 0.17 \\
Generic DPO & 0.15 & 0.19 & 0.12 & 0.14 & 0.13 \\
\textbf{DetailDPO} & \textbf{0.09} & \textbf{0.15} & \textbf{0.08} & \textbf{0.10} & \textbf{0.10} \\
\bottomrule
\end{tabular}%
}
\caption{DER by error type (lower is better). DetailDPO reduces all types, with the largest absolute gains on threshold ($\tau_1$) and unit ($\tau_2$) errors.}
\label{tab:error_analysis}
\end{table}

Threshold errors ($\tau_1$) and unit errors ($\tau_2$) show the largest absolute reductions (21 and 19 points vs.\ zero-shot), consistent with the fact that numeric tokens have the most localized representation and are thus most amenable to gradient concentration. Scope ($\tau_3$) and condition ($\tau_5$) errors, which involve longer token spans, still see substantial reductions but with smaller margins over Generic DPO, suggesting that multi-token perturbations diffuse the gradient signal compared to single-token numeric perturbations. Note that the overall DER in \Cref{tab:main_results} is a weighted average over all detail elements rather than a simple average of per-type rates, since the five types have slightly unequal prevalence in the test set.

\subsection{Multi-Model Scaling}

\begin{table}[t]
\centering
\small
\resizebox{\linewidth}{!}{%
\begin{tabular}{llcc}
\toprule
\textbf{Base Model} & \textbf{Method} & \textbf{DER$\downarrow$} & \textbf{Evid.\ F1} \\
\midrule
\multirow{2}{*}{Qwen2.5-7B} & Zero-shot & 0.28 & 0.40 \\
 & DetailDPO & \textbf{0.11} & \textbf{0.73} \\
\midrule
\multirow{2}{*}{Qwen2.5-14B} & Zero-shot & 0.22 & 0.49 \\
 & DetailDPO & \textbf{0.08} & \textbf{0.78} \\
\midrule
\multirow{2}{*}{Qwen2.5-72B} & Zero-shot & 0.14 & 0.62 \\
 & DetailDPO & \textbf{0.06} & \textbf{0.83} \\
\midrule
\multirow{2}{*}{Llama-3.1-8B} & Zero-shot & 0.31 & 0.36 \\
 & DetailDPO & \textbf{0.13} & \textbf{0.70} \\
\bottomrule
\end{tabular}%
}
\caption{DetailDPO generalizes across model scales and families.}
\label{tab:scaling}
\end{table}

\Cref{tab:scaling} shows that DetailDPO provides consistent DER reduction across all tested models. Larger models have lower baseline DER but still benefit substantially from DetailDPO (e.g., 72B: 0.14$\to$0.06, a 57\% reduction). The improvement is also consistent across model families (Qwen vs.\ Llama), indicating that the method is not architecture-specific.\looseness=-1

\subsection{Context-Length Analysis}

\begin{table}[t]
\centering
\small
\resizebox{\linewidth}{!}{%
\begin{tabular}{llccc}
\toprule
\textbf{Context Tier} & \textbf{Method} & \textbf{Compl.\ Acc.} & \textbf{DER$\downarrow$} & \textbf{Evid.\ F1} \\
\midrule
\multirow{2}{*}{Short (8K--16K)} & Zero-shot & 0.71 & 0.22 & 0.48 \\
 & DetailDPO & \textbf{0.86} & \textbf{0.08} & \textbf{0.79} \\
\midrule
\multirow{2}{*}{Medium (16K--32K)} & Zero-shot & 0.64 & 0.29 & 0.39 \\
 & DetailDPO & \textbf{0.81} & \textbf{0.11} & \textbf{0.73} \\
\midrule
\multirow{2}{*}{Long (32K--64K)} & Zero-shot & 0.58 & 0.36 & 0.31 \\
 & DetailDPO & \textbf{0.76} & \textbf{0.15} & \textbf{0.65} \\
\bottomrule
\end{tabular}%
}
\caption{Performance by context-length tier (Qwen2.5-7B). DER degrades with length for all methods, but DetailDPO maintains a large advantage.}
\label{tab:context_length}
\end{table}

As shown in \Cref{tab:context_length}, zero-shot DER increases sharply from 0.22 (Short) to 0.36 (Long), a 64\% degradation, confirming that detail faithfulness decays faster than overall accuracy as context grows. DetailDPO mitigates this degradation but does not eliminate it: DER increases from 0.08 to 0.15 as context grows, indicating that long-context detail faithfulness remains challenging even after targeted alignment. Crucially, the \emph{absolute gap} between DetailDPO and zero-shot widens at longer contexts (from 0.14 to 0.21), demonstrating that our approach is especially valuable in the long-context regime where detail hallucinations are most~prevalent.

\paragraph{Note on tier sizes.} The test set contains 1,215 Short, 519 Medium, and 266 Long samples, reflecting the combinatorial difficulty of assembling 8--15 documents into the 32K--64K range. Since all metrics in \Cref{tab:context_length} are computed \emph{per-tier}, the unequal tier sizes do not affect the reported numbers; however, the Long tier estimates have higher variance due to fewer~samples.

\subsection{Ablation Studies}

\begin{table}[t]
\centering
\small
\resizebox{\linewidth}{!}{%
\begin{tabular}{lccc}
\toprule
\textbf{Configuration} & \textbf{Compl.\ Acc.} & \textbf{DER$\downarrow$} & \textbf{Evid.\ F1} \\
\midrule
DetailDPO (Full) & \textbf{0.82} & \textbf{0.11} & \textbf{0.73} \\
\quad w/o Minimal Perturbation & 0.77 & 0.15 & 0.64 \\
\quad w/o Evidence Supervision & 0.80 & 0.12 & 0.66 \\
\quad w/o Synthetic Data & 0.79 & 0.13 & 0.70 \\
\quad w/o LoRA (full fine-tuning) & 0.81 & 0.11 & 0.72 \\
\quad SFT only (no DPO) & 0.73 & 0.19 & 0.56 \\
\bottomrule
\end{tabular}%
}
\caption{Ablation study on Qwen2.5-7B.}
\label{tab:ablation}
\end{table}

\Cref{tab:ablation} reveals the contribution of each component. The minimal detail perturbation strategy is the most important factor for DER reduction (removing it increases DER from 0.11 to 0.15), confirming the gradient concentration hypothesis. Evidence supervision primarily impacts Evidence F1 (0.73$\to$0.66) rather than DER, as expected since it targets citation quality rather than detail accuracy. Synthetic data augmentation provides a moderate boost across all metrics by increasing error-type coverage. Interestingly, LoRA achieves parity with or slightly outperforms full fine-tuning on Evidence F1 (0.73 vs.\ 0.72); we attribute this to LoRA's implicit regularization preventing overfitting on the relatively small preference dataset (10K pairs), consistent with prior observations on low-rank adaptation~\cite{hu2022lora}.

\subsection{Gradient Analysis: Empirical Verification of \Cref{prop:gradient}}
\label{sec:gradient_empirical}

To empirically verify the gradient concentration property, we compute the per-token gradient norm $\|\nabla_\theta \log \pi_\theta(y_t | y_{<t}, x)\|$ during DPO training for both detail tokens ($t \in \mathcal{P}$) and non-detail tokens ($t \in \bar{\mathcal{P}}$), averaged over 500 training samples.

For DetailDPO (minimal perturbation), the mean gradient norm at detail positions is $3.8\times$ larger than at non-detail positions ($\|\Delta_{\mathcal{P}}\| / \|\Delta_{\bar{\mathcal{P}}}\| = 3.82 \pm 0.47$). For Generic DPO (random rejected responses), this ratio drops to $1.12 \pm 0.31$, confirming that the gradient is approximately uniform across tokens. \Cref{fig:gradient_analysis} visualizes this contrast. This validates \Cref{prop:gradient} and explains DetailDPO's superior DER reduction.

\begin{figure}[t]
\centering
\includegraphics[width=0.95\columnwidth]{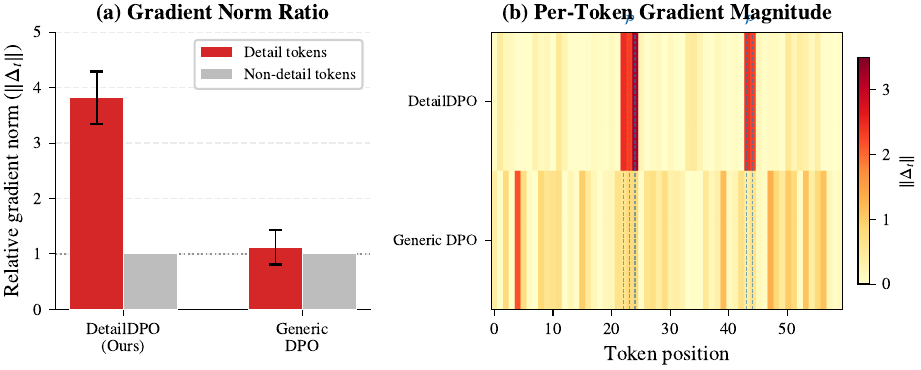}
\caption{Gradient concentration analysis. \textbf{(a)}: Mean gradient norm ratio (detail vs.\ non-detail tokens); error bars show $\pm 1$ std over 500 samples. \textbf{(b)}: Per-token gradient magnitude heatmap for a representative sample; dashed lines mark detail token positions $\mathcal{P}$. DetailDPO concentrates gradient on detail tokens, while Generic DPO distributes it uniformly.}
\label{fig:gradient_analysis}
\end{figure}

\subsection{Cross-Domain Transfer}

\begin{table}[t]
\centering
\small
\resizebox{\linewidth}{!}{%
\begin{tabular}{llcc}
\toprule
\textbf{Target Domain} & \textbf{Method} & \textbf{DER$\downarrow$} & \textbf{Evid.\ F1} \\
\midrule
\multirow{4}{*}{Financial (SEC filings)} & Zero-shot & 0.26 & 0.37 \\
 & SFT (transfer) & 0.21 & 0.48 \\
 & Generic DPO (transfer) & 0.18 & 0.52 \\
 & DetailDPO (transfer) & \textbf{0.15} & \textbf{0.58} \\
\midrule
\multirow{4}{*}{Medical (clinical guidelines)} & Zero-shot & 0.24 & 0.41 \\
 & SFT (transfer) & 0.20 & 0.50 \\
 & Generic DPO (transfer) & 0.17 & 0.54 \\
 & DetailDPO (transfer) & \textbf{0.14} & \textbf{0.61} \\
\bottomrule
\end{tabular}%
}
\caption{Cross-domain transfer: models trained on regulatory data, evaluated on out-of-domain documents (Qwen2.5-7B). All methods use the same training data; only the alignment objective differs.}
\label{tab:transfer}
\end{table}

\Cref{tab:transfer} demonstrates that DetailDPO's learned detail discrimination transfers across domains. Compared to zero-shot, DetailDPO reduces DER by 42\% on both financial and medical documents without any domain-specific fine-tuning. Importantly, the transfer advantage of DetailDPO over Generic DPO (17\% relative on financial, 18\% on medical) persists, indicating that the benefit stems specifically from the minimal perturbation strategy rather than preference optimization in general. This suggests that the detail error patterns captured by our taxonomy are domain-general.\looseness=-1

\subsection{Case Study}

We illustrate with a representative example from the Medium tier (3 documents concatenated with padding context, $\sim$25K tokens total).

\textbf{Scenario}: Assess compliance of a hydrogen storage facility operating at 75.0\,MPa.

\textbf{Zero-shot output}: States ``pressure must not exceed \underline{75.0}\,MPa'' (threshold error: the model copies the query value instead of the standard), cites wrong segment, and omits the operator certification requirement entirely.

\textbf{DetailDPO output}: Correctly extracts ``storage pressure shall not exceed \underline{70.0}\,MPa'' (from GB/T~34542 seg\_1), identifies operator certification requirement ``for facilities $>$50.0\,MPa'' (from GB/T~29729 seg\_1), and provides verbatim quotes for each citation. All five detail dimensions match the ground truth.

\section{Discussion}

\paragraph{Why does minimal perturbation work?} Our gradient analysis (\Cref{prop:gradient} and the empirical verification in \S\ref{sec:gradient_empirical}) provides a principled explanation: when the rejected response differs from the chosen response only at detail-bearing positions, the DPO gradient is forced to concentrate on exactly those positions. This is in contrast to generic DPO or RLHF, where rejected responses differ in many arbitrary ways, distributing the gradient signal across all tokens and diluting the learning signal for detail faithfulness. The analogy is to contrastive learning in vision~\cite{chen2020simclr}: hard negatives that differ in a single attribute are more informative than random~negatives.

\paragraph{Practical implications.} DetailDPO achieves a DER of 0.06 on the 72B model, meaning that on average fewer than 1 in 16 detail elements is incorrect. For compliance applications, this represents a substantial reduction in the human review burden: auditors can focus on flagged edge cases rather than re-verifying every detail. The evidence citation mechanism provides an additional safety net by enabling rapid verification of each conclusion against source~text.

\paragraph{Practical deployment considerations.} AI-assisted compliance analysis should serve as a support tool: final decisions must involve domain experts. Our evidence citation design promotes transparency and enables rapid human verification. Users should be aware that even with DetailDPO, non-zero DER means human verification remains essential for safety-critical~applications.

\section{Conclusion}

We identified \emph{detail hallucination}---subtle errors in thresholds, units, scopes, obligation levels, and conditions---as a critical yet under-studied failure mode of LLMs on long regulatory documents. To address this, we proposed a fine-grained Detail Error Taxonomy comprising five error types grounded in regulatory compliance practice, constructed \textbf{DetailBench}, a multi-source, multi-length benchmark with 13,000 human-annotated preference pairs spanning three jurisdictions, and introduced \textbf{DetailDPO}, a targeted preference optimization framework whose minimal detail perturbation strategy provably concentrates DPO gradient on detail-bearing tokens. Our experiments on the Qwen2.5 family (7B/14B/72B) and Llama-3.1-8B across three context-length tiers reveal three key findings. First, DetailDPO reduces the Detail Error Rate by 42--61\% relative to baselines, with the largest gains on threshold and unit errors where gradient concentration is most effective. Second, the advantage of DetailDPO over generic preference optimization \emph{widens} as context length increases, demonstrating that targeted alignment is especially valuable in the long-context regime where detail hallucinations are most prevalent. Third, the learned detail discrimination transfers to out-of-domain financial and medical documents without any domain-specific fine-tuning, suggesting that the detail error patterns captured by our taxonomy are domain-general.

\bibliography{references}

\newpage
\onecolumn
\appendix

\section*{Appendix}

\section{Related Work}

\subsection{Faithfulness and Hallucination in LLMs}

Hallucination in language models has been extensively studied~\cite{ji2023survey,maynez2020faithfulness}. Recent work distinguishes \emph{intrinsic} hallucination (contradicting the source) from \emph{extrinsic} hallucination (introducing unsupported claims). FActScore~\cite{min2023factscore} decomposes generation into atomic facts for fine-grained factuality evaluation. However, these frameworks evaluate factual correctness at the claim level and do not isolate \emph{detail-level} errors, i.e., subtle corruptions of numeric values, units, or conditional scopes within otherwise correct claims. Our work fills this gap by proposing a detail-specific error taxonomy and targeted~mitigation.

\subsection{Long-Context Understanding}

Scaling Transformers to long sequences has been approached through sparse attention~\cite{beltagy2020longformer,zaheer2020bigbird}, linear approximations~\cite{wang2020linformer,choromanski2021performer}, memory-efficient exact attention~\cite{dao2022flashattention}, and position encoding extensions~\cite{su2024roformer,peng2023yarn}. These advances enable processing 32K--128K+ tokens, but longer contexts introduce new faithfulness challenges: models must maintain detail accuracy across thousands of potentially relevant passages. LongBench~\cite{bai2024longbench} and SCROLLS~\cite{shaham2022scrolls} provide general long-context benchmarks; Needle-in-a-Haystack~\cite{kamradt2023needle} probes recall at specific positions. Our DetailBench specifically targets detail faithfulness degradation as context length~grows.

\subsection{Legal and Compliance NLP}

Legal AI encompasses document classification, information extraction, judgment prediction~\cite{katz2017general,zhong2018legal}, domain-adapted pretraining~\cite{chalkidis2020legal}, and benchmark construction~\cite{zheng2021casehold}. Compliance analysis, i.e., determining whether a scenario satisfies regulatory requirements, demands precise extraction of constraints including thresholds, scopes, and obligation levels from multi-document regulatory corpora. Existing work primarily addresses single-document settings and coarse compliance judgments; we focus on the under-explored problem of detail-level accuracy in multi-document compliance~extraction.

\subsection{Preference Optimization}

Direct Preference Optimization (DPO)~\cite{rafailov2023direct} eliminates the need for an explicit reward model by directly optimizing a preference objective over chosen/rejected response pairs. Extensions include KTO~\cite{ethayarajh2024kto} (Kahneman-Tversky Optimization), IPO~\cite{azar2024ipo} (Identity Preference Optimization), and SimPO~\cite{meng2024simpo} (Simple Preference Optimization), each modifying the preference loss geometry. RLHF~\cite{ouyang2022training} remains a strong alternative using PPO with a learned reward model. Prior work typically applies these methods uniformly across all tokens without distinguishing detail-bearing positions. Recent studies on hard negative mining for preference optimization have shown that the quality and structure of rejected responses significantly impact alignment effectiveness~\cite{rafailov2023direct}. Our \emph{minimal detail perturbation} strategy extends this line of work by deliberately constructing preference pairs that isolate specific detail dimensions, analogous to hard negatives that differ in exactly one attribute, enabling targeted gradient concentration on detail-bearing~tokens.

\section{Algorithm Details}

\begin{algorithm}[t]
\caption{Minimal Detail Perturbation}
\label{alg:perturbation}
\KwIn{Correct response $y_w$, error type $\tau_j$}
\KwOut{Preference pair $(y_w, y_l)$}
Identify detail elements $\{d_k : \text{type}(d_k) = \tau_j\}$ in $y_w$\;
Select target element $d_k^*$ uniformly at random\;
\Switch{$\tau_j$}{
    \Case{$\tau_1$ (threshold)}{
        $\tilde{d}_k \leftarrow d_k^* \cdot (1 + \text{Uniform}[\pm 0.1, \pm 0.2])$\;
    }
    \Case{$\tau_2$ (unit)}{
        $\tilde{d}_k \leftarrow$ \text{replace unit symbol, keep numeric value}\;
    }
    \Case{$\tau_3$ (scope)}{
        $\tilde{d}_k \leftarrow$ \text{broaden or narrow scope phrase}\;
    }
    \Case{$\tau_4$ (level)}{
        $\tilde{d}_k \leftarrow$ \text{substitute obligation keyword}\;
    }
    \Case{$\tau_5$ (condition)}{
        $\tilde{d}_k \leftarrow$ \text{alter or drop conditional clause}\;
    }
}
Construct $y_l = y_w[d_k^* \leftarrow \tilde{d}_k]$\;
\Return{$(y_w, y_l)$}
\end{algorithm}

\section{Proof of \Cref{prop:gradient} (Full Version)}

\begin{proof}
Consider a preference pair $(y_w, y_l)$ generated by minimal detail perturbation, where $\mathcal{P} \subset \{1,\ldots,T\}$ denotes the set of token positions where $y_w$ and $y_l$ differ, and $\bar{\mathcal{P}} = \{1,\ldots,T\} \setminus \mathcal{P}$.

The DPO objective involves the gradient of the reward difference:
\begin{equation}
\nabla_\theta [r_\theta(x,y_w) - r_\theta(x,y_l)] = \sum_{t=1}^{T} \Delta_t
\end{equation}
where $\Delta_t = \nabla_\theta \log \pi_\theta(y_{w,t}|y_{w,<t}, x) - \nabla_\theta \log \pi_\theta(y_{l,t}|y_{l,<t}, x)$.

\textbf{Phase 1} ($t < \min(\mathcal{P})$): By assumption (A1), $y_{w,t} = y_{l,t}$ and $y_{w,<t} = y_{l,<t}$ for all $t$ before the first perturbation, so $\Delta_t = 0$ exactly.

\textbf{Phase 2} ($t \in \mathcal{P}$): At perturbed positions, $y_{w,t} \neq y_{l,t}$. The gradient involves $\nabla_\theta \log \text{softmax}(z_t)[v]$ for different tokens $v$. Since the logit gradients for different vocabulary items point in different directions in parameter space, $\|\Delta_t\| \geq \delta > 0$ as guaranteed by assumption (A3).

\textbf{Phase 3} ($t \in \bar{\mathcal{P}}, t > \min(\mathcal{P})$): Here $y_{w,t} = y_{l,t}$ but the conditioning contexts $y_{w,<t}$ and $y_{l,<t}$ differ through the perturbed tokens. We justify assumption (A2) as follows. Let $h_t^w$ and $h_t^l$ denote the hidden states at position $t$ under the two contexts. The softmax-attention mechanism is locally Lipschitz when the model weights $W$ satisfy $\|W\| \leq B$ (bounded spectral norm) and the input embeddings lie in a compact set, conditions that hold in practice due to weight decay and layer normalization. Under these conditions, $\|h_t^w - h_t^l\| \leq L \cdot \|h_{\mathcal{P}}^w - h_{\mathcal{P}}^l\|$ for a Lipschitz constant $L$ that depends on the attention weights attending to perturbed positions. The per-token gradient difference is then bounded: $\|\Delta_t\| \leq \epsilon$ where $\epsilon$ depends on $L$, $B$, and the attention mass allocated to $\mathcal{P}$ from position $t$.

Summing over the three phases:
\[
\left\| \sum_{t \in \bar{\mathcal{P}}} \Delta_t \right\| \leq |\bar{\mathcal{P}}| \cdot \epsilon, \qquad \left\| \sum_{t \in \mathcal{P}} \Delta_t \right\| \geq |\mathcal{P}| \cdot \delta
\]
where the lower bound on $\mathcal{P}$ assumes aligned gradients (worst-case analysis without alignment yields $\sqrt{|\mathcal{P}|} \cdot \delta$, which still dominates the non-detail contribution for small $\epsilon$). For typical regulatory documents with $|\mathcal{P}| / T \approx 0.01$--$0.05$, the gradient is dominated by the $|\mathcal{P}|$ perturbed positions when $|\bar{\mathcal{P}}| \cdot \epsilon < |\mathcal{P}| \cdot \delta$.
\end{proof}

\section{DetailBench Annotation Guidelines}

Annotators follow a three-step protocol: (1) \textbf{Identify detail elements}: For each compliance query, mark all detail-bearing spans in the source documents (thresholds, units, scope phrases, obligation keywords, conditional clauses). (2) \textbf{Classify error type}: Assign each element to one of $\tau_1$--$\tau_5$. (3) \textbf{Record ground truth}: Extract the correct value and link it to the source segment with a verbatim quote. Disagreements are resolved by majority vote among three annotators. We release the full annotation manual with the benchmark.

\section{Complete Task Example}

We provide a complete input/output example demonstrating the compliance analysis task on a multi-document input.

\textbf{Scenario:} A hydrogen storage facility operates at 75.0\,MPa with leak detection and emergency shutdown systems installed. Operators have basic safety training but lack high-pressure certification. Daily processing: 500\,kg.

\textbf{Input:} 3 regulatory documents (GB/T 34542-2023: storage safety; GB/T 29729-2023: operations; GB/T 24499-2023: environmental), totaling $\sim$2K tokens of regulatory text, embedded in a context of $\sim$25K tokens (Medium tier).

\textbf{DetailDPO Output} (abridged):
\begin{Verbatim}[frame=single, framesep=2mm, fontsize=\footnotesize, breaklines=true]
{
  "is_compliant": false,
  "key_constraints": [
    {
      "constraint": "Storage pressure <= 70.0 MPa",
      "type": "threshold",
      "current": "75.0 MPa",
      "violation": true,
      "evidence": {
        "doc": "GB_T_34542_2023",
        "seg": "seg_1",
        "quote": "shall not exceed 70.0 MPa"
      }
    },
    {
      "constraint": "Operator certification required (>50 MPa)",
      "type": "level",
      "current": "basic training only",
      "violation": true,
      "evidence": {
        "doc": "GB_T_29729_2023",
        "seg": "seg_1",
        "quote": "must hold valid certification"
      }
    }
  ],
  "risks": [
    "Pressure 7.1% over limit (high severity)",
    "Uncertified operators (high severity)"
  ]
}
\end{Verbatim}

All detail parameters (70.0\,MPa threshold, 50.0\,MPa condition, obligation levels) match ground truth exactly. The zero-shot baseline incorrectly states ``75.0\,MPa'' as the threshold (copying the query value) and misses the 50.0\,MPa conditional scope.

\section{Training Details}
\label{sec:training_details}

\subsection{LoRA Configuration}

LoRA adapters are applied to attention layers (q\_proj, k\_proj, v\_proj, o\_proj) and feed-forward layers (gate\_proj, up\_proj, down\_proj) with rank $r=16$, $\alpha=32$, and dropout $0.1$.

\subsection{Computational Resources}

Training is performed using BF16 mixed precision with FlashAttention-2~\cite{dao2022flashattention} for memory-efficient long-context training. Training on Qwen2.5-7B uses 4$\times$A100-80GB, taking $\sim$18 hours for 10K preference pairs. The 14B model uses the same 4$\times$A100-80GB setup, taking $\sim$30 hours. The 72B model uses 8$\times$A100-80GB with DeepSpeed ZeRO-3, taking $\sim$72 hours. Llama-3.1-8B uses 4$\times$A100-80GB, taking $\sim$20 hours. All RLHF baselines use identical compute budgets for fair comparison.

\section{Human Evaluation Protocol}

Three domain experts independently evaluate 200 randomly sampled outputs on three dimensions: (1) compliance judgment correctness, (2) detail parameter accuracy (per-element), and (3) evidence citation relevance. Inter-annotator agreement: Cohen's $\kappa = 0.81$ for judgment, 0.76 for detail accuracy, 0.83 for citation relevance.

\section{Baseline Descriptions}
\label{sec:baselines_appendix}

\begin{itemize}
    \item \textbf{Zero-shot}: Base model without any fine-tuning, prompted directly with the compliance query and document context.
    \item \textbf{SFT}: Supervised fine-tuning on compliance conclusions without preference data.
    \item \textbf{RLHF}: PPO~\cite{schulman2017ppo} with a reward model trained on the same preference data (1,024 reward model training steps; the reward model uses the same architecture as the policy).
    \item \textbf{Generic DPO}: Standard DPO with randomly generated rejected responses (not minimal perturbations).
    \item \textbf{KTO}~\cite{ethayarajh2024kto}: Kahneman-Tversky Optimization, a DPO variant using per-example losses derived from prospect theory.
    \item \textbf{SimPO}~\cite{meng2024simpo}: Simple Preference Optimization with length-normalized rewards and no reference model.
    \item \textbf{RAG + Zero-shot}: BM25 retrieval (top-$k{=}5$ passages per query) of relevant segments, concatenated as context for zero-shot LLM generation. We use BM25 rather than DPR as it requires no in-domain training and provides a strong sparse-retrieval baseline.
    \item \textbf{GPT-4o}: OpenAI commercial API baseline (zero-shot, temperature=0) representing frontier general-purpose LLM capability.
    \item \textbf{Claude-3.5-Sonnet}: Anthropic commercial API baseline (zero-shot, temperature=0).
\end{itemize}

All fine-tuning baselines (SFT, RLHF, Generic DPO, KTO, SimPO) use the same training data as DetailDPO; only the alignment objective differs.

\section{Per-Type Failure Case Comparison}
\label{sec:failure_cases}

\Cref{tab:failure_cases} illustrates representative zero-shot failures and the corresponding DetailDPO corrections for each error type.

\begin{table*}[t]
\centering
\small
\begin{tabular}{p{0.06\linewidth}p{0.20\linewidth}p{0.30\linewidth}p{0.30\linewidth}}
\toprule
\textbf{Type} & \textbf{Ground Truth} & \textbf{Zero-shot Output (Error)} & \textbf{DetailDPO Output (Correct)} \\
\midrule
$\tau_1$ & Storage pressure $\leq$ \textbf{70.0}\,MPa (GB/T~34542, seg\_1) & ``pressure must not exceed \textbf{75.0}\,MPa'' \newline \emph{Copied query value instead of standard} & ``storage pressure shall not exceed \textbf{70.0}\,MPa'' \newline [Evidence: GB/T~34542 seg\_1] \\
\midrule
$\tau_2$ & Daily capacity: 500\,\textbf{kg} (GB/T~29729, seg\_3) & ``daily processing capacity of 500\,\textbf{m\textsuperscript{3}}'' \newline \emph{Swapped unit without conversion} & ``daily processing: 500\,\textbf{kg}'' \newline [Evidence: GB/T~29729 seg\_3] \\
\midrule
$\tau_3$ & Applies to \textbf{stationary} storage only (CFR~49.178, seg\_2) & ``applies to \textbf{all hydrogen storage} systems'' \newline \emph{Broadened scope to include mobile} & ``applies to \textbf{stationary} storage facilities'' \newline [Evidence: CFR~49.178 seg\_2] \\
\midrule
$\tau_4$ & Operators \textbf{shall} hold certification (GB/T~29729, seg\_1) & ``operators \textbf{should} receive training'' \newline \emph{Downgraded mandatory to advisory} & ``operators \textbf{shall} hold valid certification'' \newline [Evidence: GB/T~29729 seg\_1] \\
\midrule
$\tau_5$ & Cert.\ required \textbf{when $>$ 50\,MPa} (GB/T~29729, seg\_1) & ``certification is required for operators'' \newline \emph{Dropped pressure condition entirely} & ``certification required \textbf{for facilities $>$ 50.0\,MPa}'' \newline [Evidence: GB/T~29729 seg\_1] \\
\bottomrule
\end{tabular}
\caption{Representative failure cases: zero-shot vs.\ DetailDPO outputs for each error type. Bold marks the critical detail element. Zero-shot errors are subtle (surface-plausible) but safety-critical; DetailDPO corrects them with proper evidence citations.}
\label{tab:failure_cases}
\end{table*}

% \subsection{Code and Data}

% Code and data are open-sourced:
% \begin{itemize}
%     \item GitHub repository: \url{https://github.com/xxx/LongContextSC}
%     \item Dataset: \url{https://huggingface.co/datasets/xxx/hydrogen-compliance}
% \end{itemize}

\end{document}